\documentclass[letterpaper, 10 pt, conference]{ieeeconf}  

\IEEEoverridecommandlockouts                              

\usepackage{graphics} 
\usepackage{epsfig} 
\usepackage{mathptmx} 
\usepackage{amsmath} 
\usepackage{amssymb}  
\usepackage{algorithm}
\usepackage{algpseudocode}

\algnewcommand\algorithmicinput{\textbf{Offline:}}
\algnewcommand\Offline{\item[\algorithmicinput]}

\algnewcommand\algorithmicoutput{\textbf{Online:}}
\algnewcommand\Online{\item[\algorithmicoutput]}

\usepackage[dvipsnames]{xcolor}
\usepackage{arydshln}

\usepackage[english]{babel}
\usepackage{hyperref}
\hypersetup{
    colorlinks=true,
    linkcolor=blue,
    filecolor=magenta,  
    urlcolor=cyan
    }
    
\newtheorem{theo}{Theorem}

\newtheorem{assu}{Assumption}

\newtheorem{lemma}{Lemma}

\newtheorem{coro}{Corollary}

\newtheorem{remark}{Remark}

\title{\LARGE \bf
Adaptive Output Feedback Model Predictive Control
}

\author{Anchita Dey, Abhishek Dhar and Shubhendu Bhasin
\thanks{Anchita Dey \textit{(corresponding author)} and Shubhendu Bhasin  are with the Department of Electrical Engineering, Indian Institute of Technology Delhi, Hauz Khas, New Delhi, Delhi 110016, India {\tt\small \{anchita.dey, sbhasin\}@ee.iitd.ac.in.}}%
\thanks{Abhishek Dhar is with the Department of Electrical Engineering, Link{\"o}ping University, Link{\"o}ping 58183, Sweden 
        {\tt\small abhishek.dharr@gmail.com.}}%
}

\begin{document}

\maketitle
\thispagestyle{empty}
\pagestyle{empty}

\begin{abstract}
Model predictive control (MPC) for uncertain systems in the presence of hard constraints on state and input is a non-trivial problem, and the challenge is increased manyfold in the absence of state measurements. In this paper, we propose an adaptive output feedback MPC technique, based on a novel combination of an adaptive observer and robust MPC, for single-input single-output discrete-time linear time-invariant systems. At each time instant, the adaptive observer provides estimates of the states and the system parameters that are then leveraged in the MPC optimization routine while robustly accounting for the estimation errors. The solution to the optimization problem results in a homothetic tube where the state estimate trajectory lies. The true state evolves inside a larger outer tube obtained by augmenting a set, invariant to the state estimation error, around the homothetic tube sections. The proof for recursive feasibility for the proposed `homothetic and invariant' two-tube approach is provided, along with simulation results on an academic system.
\end{abstract}

\section{Introduction}

\noindent Model predictive control (MPC) \cite{kouvaritakis2016model} is a well-known technique for making optimal control decisions in the presence of hard constraints on states and control inputs. The classical MPC relies on exact system model knowledge without consideration for any uncertainty or noise. Unfortunately, it is practically challenging to obtain an approximate, let alone an accurate model of the system. Thus, there are efforts to relax the need for `M' in MPC. Several questions, however, arise. Can state predictions, which are at the heart of MPC, be reliably made without an accurate model? How to account for errors in those state predictions? Is recursive feasibility guaranteed? These questions become harder to answer when the state measurements are also unavailable.

There has already been a lot of progress on how to deal with uncertainties in the system model, mostly using a robust tube-based approach \cite{langson2004robust}, \cite{rakovic2012homothetic}. However, robust approaches typically handle additive disturbances. To account for parametric uncertainties in the model, adaptive MPC methods were proposed, where the goal was to learn the unknown parameters for better transient performance \cite{lorenzen2019robust}\nocite{hernandez2016persistently}\nocite{dhar2021indirect}-\cite{zhu2019constrained}. Another learning-based approach is the iterative learning-based MPC \cite{rosolia2017learning}, \cite{bujarbaruah2018adaptive} where learning is done to increase the accessible safe region for the states to evolve. This method is particularly useful for applications involving repetitive tasks.

A restrictive assumption made in most robust and adaptive MPC literature is the availability of state measurements. To obviate this requirement, \cite{mayne2006robust} and \cite{kogel2017robust}  use a Luenberger observer to estimate the states, and then design robust tubes for a system with known model parameters but additive disturbances. In \cite{subramanian2017novel}, a tube-based approach independent of the state estimation method is presented. The application-based paper \cite{ghanes2016robust} uses MPC with an observer for state estimation; the adaptation done is in terms of switching from one known system to another. Authors of \cite{brunke2021rlo} extend the work in \cite{mayne2006robust} to learning of the safe region for iterative tasks, as in \cite{rosolia2017learning}, while assuming the system parameters to be known. \textcolor{black}{Recently, data-driven techniques \cite{coulson2019data, berberich2020data} have been proposed to solve the robust MPC problem. The result hinges on learning an implicit non-parametric model (or behavior) of the system characterized by a Hankel matrix, which is obtained from sufficiently rich data and must be constructed prior to running the MPC optimization routine.} A drastically different method is adopted in \cite{nguyen2020output}, where input-output data is used to estimate future outputs with an estimated ARX model, whose parameters, along with the input to be applied, are obtained by solving recursive least square (RLS) problems. 

Yet, a common assumption in  \cite{mayne2006robust}-\cite{brunke2021rlo} is that the system parameters are known, whereas the techniques in \cite{coulson2019data}, \cite{berberich2020data} need a pre-MPC data collection phase for Hankel matrix construction. Further, since the Hankel matrix is never updated, it is unclear if \cite{coulson2019data, berberich2020data} are robust to small perturbations in the system dynamics. The approach in \cite{nguyen2020output} uses RLS with saturated control inputs, which may cause violation of any hard constraints on the outputs.

In this paper, we strategically combine an online state and parameter estimator with MPC leading to an adaptive output feedback MPC (AOFMPC) design for discrete-time single-input single-output (SISO) linear time-invariant (LTI) systems. The estimator deployed is an adaptive observer motivated from \cite{suzuki1980discrete}, \cite{kudva1974discrete}, which uses input-output data to simultaneously estimate the states and the system parameters \textit{online}. The state estimation error is accounted for by constructing an invariant set \cite{rakovic2005invariant} that is used to get a tightened constraint for the state estimates. The constrained optimal control problem (COCP) for MPC is reformulated in terms of the state and parameter estimates, i.e., the cost function along with the constraints are now defined in terms of the state and parameter estimates obtained from the observer. Another challenge is to generate future state estimate predictions required in the COCP by somehow exploiting the adaptive observer dynamics; this is non-trivial since it requires knowledge of the future parameter estimates and outputs. To get around this issue, the observer dynamics is rewritten in terms of the quantities known at the current time instant, and an additive uncertainty is considered to account for the mismatched dynamics. This uncertainty in predicted state estimates is dealt with using a homothetic tube framework \cite[sec~3.3]{langson2004robust} that ensures that the state estimate trajectory stays inside the homothetic tube. The set invariant to the state estimation errors is added to the homothetic tube sections to obtain a larger outer tube, inside which the true state evolves. This leads to a unique two-tube architecture: an inner homothetic tube and an outer invariant-added-to-homothetic tube.

The novel framework of an invariant annular portion around a homothetic tube enables us to relax the assumption on availability of state measurements and system parameters. In addition, unlike \cite{coulson2019data}-\cite{nguyen2020output}, \cite{heirung2017dual} that use input and output constraints, the use of a state estimator in the proposed framework allows imposition of user-defined constraints on all the states and inputs. 
   

\textit{Notations: }$||\cdot||_q$ denotes $q$-norm of a vector 
where $q\in\{2,\infty\}$. \textcolor{black}{A matrix $V\succ 0$ implies $V$ is symmetric positive-definite. For a vector $g$ and a matrix $V\succ 0$, $||g||^2_V\triangleq g^TVg$.  For sets $\mathbb{A}$ and $\mathbb{B}$, Minkowski sum $\mathbb{A}\oplus\mathbb{B}\triangleq \{g+h|\;g\in\mathbb{A},\;h\in\mathbb{B}\}$, Pontryagin difference $\mathbb{A}\ominus\mathbb{B}\triangleq\{g|\;g+h\in\mathbb{A}\;\forall h\in\mathbb{B}\}$, and $conv(\mathbb{A})$ denotes convex hull of all elements in $\mathbb{A}$.} $\mathbb{I}_g^h\triangleq\{g,g+1,...,h-1,h\}$ with integers $g$ and $h>g$, $I_g\in\mathbb{R}^{g\times g}$ is the Identity matrix and $0_{g\times h}\in\mathbb{R}^{g\times h}$ is the zero matrix. The value of $g$ at time $t+i$ predicted at time $t$ is denoted by $g_{i|t}$. $g\in \mathcal{L}_\infty$ implies $g$ is bounded.

\section{Problem Formulation}
Consider the constrained discrete-time SISO LTI system
\begin{align}
    &x_{t+1}=Ax_t+bu_t\text{ , }\;\;y_t=cx_t \label{osys1}\\
    &x_t\in\mathcal{X}\text{ , }\;u_t\in\mathcal{U} \hspace{0.4cm}\forall t\in\mathbb{I}_0^\infty \label{hc1}
\end{align}
where $x_t\in\mathbb{R}^n$, $u_t\in\mathbb{R}$ and $y_t\in\mathbb{R}$ denote the state, input and output, respectively at time $t$. $\mathcal{X}$and $\mathcal{U}$ are known compact, convex polytopes containing their respective origins. The parameters $A\in\mathbb{R}^{n\times n}$ and $b,c^T\in\mathbb{R}^n$ are unknown constants. The goal is to stabilize the system while satisfying hard state and input constraints (\ref{hc1}). Provided the state measurements and system parameters are available, this is achievable with classical MPC \cite{kouvaritakis2016model} by solving the following.
\begin{equation*}
\mathbb{OP}1:\text{ }\min_{\mu_t}\text{ }J(x_t,\mu_t)\triangleq \sum_{i=0}^{N-1}\Big{(}||x_{i|t}||^2_{{Q}}+{{r}}u_{i|t}^2\Big{)}+||x_{N|t}||^2_{\bar{P}}
\end{equation*}
\begin{equation*}
\text{subject to }x_{0|t}=x_t\text{, }x_{i|t}\in\mathcal{X}\hspace{0.2cm}\forall i\in\mathbb{I}_0^{N}\text{, } u_{i|t}\in\mathcal{U}\hspace{0.2cm}\forall i\in\mathbb{I}_{0}^{N-1}\text{, }
\end{equation*}
\begin{equation*}
   x_{N|t}\in\mathcal{X}_T\subseteq\mathcal{X}\text{ and }x_{i+1|t}=Ax_{i|t}+bu_{i|t}\hspace{0.2cm}\forall i\in\mathbb{I}_{0}^{N-1}
\end{equation*}
where $\mu_t\triangleq\{u_{0|t}, u_{1|t},...,u_{N-1|t}\}$, $N$ is the prediction horizon length (for states as well as control), $Q$, $\bar{P}\in\mathbb{R}^{n\times n}$, $r\in\mathbb{R}$ with ${Q}$, $\bar{P}\succ 0$, ${r}>0$  and $\mathcal{X}_T$ is the terminal set that contains the origin and is related to $\bar{P}$ \cite[ch.~2]{kouvaritakis2016model}. 

In the absence of state measurements and accurate knowledge of $A$ and $b$, $\mathbb{OP}1$ is not solvable. In this work, an adaptive observer is used to simultaneously estimate the state and system parameters. The estimates are used to solve a reformulated COCP in a robust framework with a receding horizon approach \cite{kouvaritakis2016model}. In developing the theory, we consider\newpage\noindent the following standard assumptions \cite{lorenzen2019robust}, \cite{dhar2021indirect}, \cite{mayne2006robust} throughout.
\begin{assu}\label{ACa}
The state space realization given in \eqref{osys1} is observable.
\end{assu}
\begin{assu}\label{ABa}
The unknown parameter $\psi\triangleq\begin{bmatrix}A&b\end{bmatrix}\in\mathbb{R}^{n\times (n+1)}$ belongs to a set $\Psi\triangleq conv(\{\psi^{vi}\;\;|\;\; i\in\mathbb{I}_1^L\})$ where the vertices  $\psi^{v1},\psi^{v2},...,\psi^{vL}\in\mathbb{R}^{n\times (n+1)}$ are known, and $L$ is some finite positive integer. Each element in $\Psi$ forms a stabilizable system.
\end{assu}

Assumption \ref{ACa} is necessary for an \textcolor{black}{adaptive} observer-based approach \cite{suzuki1980discrete} \textcolor{black}{which hinges on the observable canonical form} 
while Assumption \ref{ABa} is considered to allow constraint tightening as well as to find a terminal set in the MPC design.
\section{Adaptive Observer for State and Parameter Estimation}
Following Assumption \ref{ACa}, there is no loss of generality if \eqref{osys1} is considered to be in the observable canonical form with
\begin{align}
    &A=\left[
    \begin{array}{c|c}
    a&\begin{array}{cc}
         I_{n-1}\\0_{1\times(n-1)}
    \end{array}
    \end{array}
    \right]\text{ where }a=\begin{bmatrix}
a_1\;\;a_2\;\;...\;\;a_n
\end{bmatrix}^T_\text{,} \label{A1}\\
   & b=\begin{bmatrix}
    b_1&b_2&...&b_n
    \end{bmatrix}^T\text{ and }c=\begin{bmatrix}1&0_{1\times(n-1)}\end{bmatrix}.\label{B1}
\end{align}
The structures in \eqref{A1} and \eqref{B1} are used to construct the observer. Let $F\in\mathbb{R}^{n\times n}$ be any user-defined Schur stable matrix of the form
\begin{equation}\label{F1}
F=\left[
    \begin{array}{c|c}
    f&\begin{array}{cc}
         I_{n-1}\\0_{1\times(n-1)}
    \end{array}
    \end{array}
    \right]\text{ with }f=\begin{bmatrix}
f_1&f_2&...&f_n
\end{bmatrix}^T.
\end{equation}
Using \eqref{A1}-\eqref{F1}, the plant dynamics in \eqref{osys1} can be re-written as
\begin{equation}\label{osys1r}
    x_{t+1}=Fx_t+(A-F)x_t+bu_t=Fx_t+(a-f)y_t+bu_t
\end{equation}
and an adaptive observer is designed with the dynamics \cite{kudva1974discrete}
\begin{equation}
    \label{AOsys1}
    \hat{x}_{t+1}=F\hat{x}_t+(\hat{a}_{t\textcolor{black}{+1}}-f)y_t+\hat{b}_{t\textcolor{black}{+1}}u_t.
\end{equation}
where $\hat{x}_t$, $\hat{a}_t\triangleq [\hat{a}_{1_t}\hspace{0.15cm}\hat{a}_{2_t}\hspace{0.15cm}...\hspace{0.15cm}\hat{a}_{n_t}]^T$, $\hat{b}_t\triangleq [\hat{b}_{1_t}\hspace{0.15cm}\hat{b}_{2_t}\hspace{0.15cm}...\hspace{0.15cm}\hat{b}_{n_t}]^T$ $\in\mathbb{R}^n$ are the estimates of $x_t$, $a$ and $b$, respectively. 
From \eqref{osys1r} and \eqref{AOsys1}, $y_t$ and the adaptive observer output $\hat{y}_t$ are respectively given by
\begin{align}
y_{t+1}=\phi_{t}^Tp+\textcolor{black}{cFx_{t}}\text{ and }
    \hat{y}_{t+1}=\phi_{t}^T\hat{p}_{t+1}+\textcolor{black}{cF\hat{x}_{t}}
\end{align}
where $p\triangleq\begin{bmatrix}(a-f)^T&b^T\end{bmatrix}^T\in\mathbb{R}^{2n}$, $\phi_t\triangleq [I_ny_t\;\;\;I_nu_t]^Tc^T\in\mathbb{R}^{2n}$ and $\hat{p}_t\triangleq\begin{bmatrix}(\hat{a}_t-f)^T&\hat{b}_t^T\end{bmatrix}^T\in\mathbb{R}^{2n}$ are the parameter vector, measurable regressor and the estimate of $p$, respectively.


Let $\tilde{x}_t\triangleq x_t-\hat{x}_t\in\mathbb{R}^n$ be the state estimation error and $\tilde{p}_t\triangleq p-\hat{p}_t\in\mathbb{R}^{2n}$ be the parameter estimation error. From \eqref{osys1r} and \eqref{AOsys1}, we can write
\begin{equation}\label{xtilde}
    \tilde{x}_{t+1}=F\tilde{x}_t+[I_ny_t\;\;\;I_nu_t]\tilde{p}_{t+1}.
\end{equation}
 Since \textcolor{black}{$\tilde{x}_0$ is finite and $F$ is Schur stable}, designing a suitable adaptation law so that $\tilde{p}_t$ converges to zero as $t\rightarrow\infty$ guarantees the convergence of $\tilde{x}_t$ to zero as $t\rightarrow\infty$ \cite{kudva1974discrete}. The following adaptive law is considered for updating $\hat{p}_t$ using an exponentially weighted least squares method \cite[ch.~3]{goodwin2014adaptive}
\begin{equation}\label{rec01}
    \hat{p}_{t}=\hat{p}_{t-1}+\Gamma_{t}\phi_{t-1}(y_{t}-\textcolor{black}{cF\hat{x}_{t-1}}-\phi_{t-1}^T\hat{p}_{t-1})
\end{equation}
\begin{equation}\label{rec02}
\Gamma_{t}=\frac{1}{\zeta^2}\bigg{[}\Gamma_{t-1}-\frac{\Gamma_{t-1}\phi_{t-1}\phi_{t-1}^T\Gamma_{t-1}}{\zeta^2+\phi_{t-1}^T\Gamma_{t-1}\phi_{t-1}}\bigg{]}
\end{equation}
where $\zeta\in(0,1)$ and $\Gamma_0=\gamma^2I_{2n}$ where $\gamma\gg1$ \cite{suzuki1980discrete}. The estimate $\hat{\psi}_t\triangleq\begin{bmatrix}\hat{A}_t&\hat{b}_t\end{bmatrix}$ (where $\hat{A}_t$ is structurally similar to\newpage\noindent $A$, except for $a$ being replaced by $\hat{a}_t$), is constructed using $\hat{p}_t$ given by \eqref{rec01}. 
\begin{lemma}\label{L01}
If the regressor $\phi_t$  is persistently exciting, then the parameter estimation error $\tilde{p}_t$ and consequently the state 
estimation error $\tilde{x}_t$ converges to zero as $t\rightarrow\infty$.\end{lemma}
\begin{proof}In \cite[sec.~II]{kudva1974discrete}, which is directly applicable here since the regressor $\phi_t$ is measurable. \end{proof}
\begin{remark}\label{ronuSR}
For an LTI system, the persistent excitation condition is guaranteed by choosing the input $u_t$ to be sufficiently rich \cite{boyd1986necessary}.
\end{remark}

The parameter estimates $\hat{\psi}_t$ obtained using \eqref{rec01} and \eqref{rec02} may lie outside $\Psi$. Since it is desirable that $\hat{\psi}_t\in\Psi$ (Assumption \ref{ABa}), the adaptation law in \eqref{rec01} is replaced by the following two equations
\begin{align}
 & \bar{p}_{t}\triangleq\hat{p}_{t-1}+\Gamma_{t}\phi_{t-1}(y_{t}-cF\hat{x}_{t-1}-\phi_{t-1}^T\hat{p}_{t-1})  \label{rec1}\\
 &   \hat{p}_{t}=\begin{cases}\bar{p}_{t}\text{, if }\hat{\psi}_{t}\in\Psi  \label{Proj}\\
    \text{Proj}_{\Pi}(\bar{p}_{t}) \text{, otherwise}
    \end{cases}
\end{align}
where \textcolor{black}{$\text{Proj}_{\Pi}(\bar{p}_{t})\triangleq \text{arg} \min_{\rho\in \Pi} ||\bar{p}_t-\rho||_2$}, and the set $\Pi\triangleq\big{\{}\begin{bmatrix}(\hat{a}-f)^T&\hat{b}^T\end{bmatrix}^T|\;\hat{a},\hat{b}\in\mathbb{R}^{n}\text{, }\scriptsize{\left[
    \begin{array}{c|c|c}
    \hat{a}&\begin{array}{cc}
         I_{n-1}\\0_{1\times(n-1)}
    \end{array}&\hat{b}
    \end{array}
    \right]}\in\Psi \Big{\}}$ is convex with known $L$ vertices corresponding to those of $\Psi$. 

\begin{remark}\label{L1}
Lemma \ref{L01} is applicable to the projection modified recursive updates of $\hat{p}_t$ given by \eqref{rec02}-\eqref{Proj} \cite[ch.~4]{ioannou2006adaptive}.
\end{remark}
\section{Adaptive Output Feedback MPC}
In absence of state measurements, the COCP $\mathbb{OP}1$ is reformulated in terms of the state estimates. This implies that we need to find a constraint set for $\hat{x}_t$, so that $x_t\in\mathcal{X}$. Since $\hat{x}_t=x_t-\tilde{x}_t$, it is possible to obtain a constraint set $\widehat{\mathcal{X}}$ for $\hat{x}_t$ by tightening $\mathcal{X}$ using an invariant set for $\tilde{x}_t$ \cite{rakovic2005invariant}. 
\subsection{An invariant set for the state estimation errors}
 From \eqref{AOsys1}, we can write the adaptive observer dynamics as
\begin{equation}\label{AOdy}
    \hat{x}_{t+1}=F\hat{x}_t+(\hat{A}_{t\color{black}{+1}}-F){x}_t+\hat{b}_{t\color{black}{+1}}u_t.
\end{equation}
Subtracting (\ref{AOdy}) from (\ref{osys1r}), we get
\begin{equation}\label{OSEdy}
    \tilde{x}_{t+1}=F\tilde{x}_t+\tilde{\psi}_{t\color{black}{+1}}Z_t
\end{equation}
where $\tilde{\psi}_t\triangleq\psi-\hat{\psi}_t\in\mathbb{R}^{n\times (n+1)}$ and $Z_t\triangleq\begin{bmatrix}x_t^T&u_t\end{bmatrix}^T\in\mathbb{R}^{n+1}$. Next, we make a standard assumption regarding the uncertainty in initial state estimate \cite{mayne2006robust}, \cite{kogel2017robust}.
\begin{assu}\label{Wa}
There exists a known, convex and compact set $\mathbb{W}_0$ containing $\tilde{x}_0$ and the origin,  \textcolor{black}{that satisfies $Fw\in\mathbb{W}_0$ $\forall w\in\mathbb{W}_0$ (positive invariance)}, and $\mathbb{W}_0\subset \mathcal{X}$.
\end{assu}
Assumption \ref{Wa} is required to characterize a bounded set for $\tilde{x}_t$ $\forall t\in\mathbb{I}_0^\infty$. Using it and (\ref{OSEdy}), we can write
\begin{equation}\label{Infsum}
   \tilde{x}_0\in\mathbb{W}_0\text{ and } \tilde{x}_i\in\mathbb{W}_i\triangleq F^i \mathbb{W}_0 \oplus \bigoplus_{k=0}^{i-1} F^k \Omega_1 \hspace{0.4cm}\forall i\in\mathbb{I}_1^\infty
\end{equation}
where \textcolor{black}{$\Omega_1 \triangleq\{\bar{\psi}Z\;|\; \bar{\psi}\in\widetilde{\Psi} \text{; }Z=[\;x^T\;u\;]^T$ $\text{with }x\in\mathcal{X},\;u\in\mathcal{U}\}$} and \textcolor{black}{$\widetilde{\Psi}\triangleq conv(\{\psi^{vi}-\psi^{vj}\;|\;i\neq j\; \forall i,j\in\mathbb{I}_1^L\})$}. Let $\Omega_\infty\triangleq \lim_{i\rightarrow\infty} \bigoplus_{k=0}^{i-1} F^k \Omega_1$. Since $\Omega_1$ is a compact, convex set containing the origin, we can find a minimal robust positively invariant (RPI) set $\Omega_\infty$ that satisfies $\Omega_\infty=F\Omega_\infty\oplus\Omega_1$ \cite{kogel2017robust}. Computing a \textcolor{black}{RPI} outer approximation \cite{rakovic2005invariant} of $\Omega_\infty$ is more tractable; let $\Omega_{oa}$ be that outer approximation. Using $\Omega_{oa}$ and the fact that $\mathbb{W}_0\supseteq F\mathbb{W}_0 \supseteq F^2\mathbb{W}_0\supseteq...$ (follows from Assu-\newpage\noindent mption \ref{Wa}), we get
\begin{align}\label{xtinW0}
    \tilde{x}_t\in\mathbb{W}_{oa}\triangleq \mathbb{W}_0\oplus \Omega_{oa}\hspace{0.7cm}\forall t\in\mathbb{I}_0^\infty.
\end{align}
 
Finally, using $\mathbb{W}_{oa}$ in which the state estimation error $\tilde{x}_t$ evolves, the tightened constraint set for the state estimate $\hat{x}_t$ is obtained as \begin{equation}\label{hatXtc}
x_t-\tilde{x}_t=\hat{x}_t\in\widehat{\mathcal{X}}\triangleq\mathcal{X}\ominus\mathbb{W}_{oa}.
\end{equation}

\begin{remark}\textcolor{black}{It is implied from \eqref{hatXtc} that smaller the size of $\mathbb{W}_{oa}$, the larger the feasible region for the COCP.}\end{remark}
%

The initial constraint tightening in \eqref{hatXtc} allows us to reformulate the COCP in terms of the state estimate. Adding back $\mathbb{W}_{oa}$ to the state estimate trajectory, obtained from the subsequently reformulated COCP, results in a tube \cite{langson2004robust} for the true state trajectory.

The observer dynamics in \eqref{AOsys1}, in addition to generating state estimates, is exploited for computing state estimate predictions that are used in the MPC optimization routine. To that end, the dynamics in \eqref{AOdy}, which is an equivalent form of \eqref{AOsys1}, is rewritten as
\begin{equation}\label{epsil1}
    \hat{x}_{t+i+1}=F\hat{x}_{t+i}+(\hat{A}_{t+i\color{black}{+1}}-F){x}_{t+i}+\hat{b}_{t+i\color{black}{+1}}u_{t+i}\hspace{0.1cm}\forall i\in\mathbb{I}_0^\infty.
\end{equation}
%
However, in its current form, \eqref{epsil1} is not usable for computing state estimate predictions due to the presence of $\hat{A}_{t+i\color{black}{+1}}$, $\hat{b}_{t+i\color{black}{+1}}$ and $x_{t+i}$ $\forall i\in\mathbb{I}_0^\infty$, which are unavailable at current time $t$. Hence, \eqref{epsil1} is rewritten as
\begin{equation}\label{epsil2}
    \hat{x}_{t+i+1}=\hat{A}_t\hat{x}_{t+i}+\hat{b}_tu_{t+i}+\varepsilon_{t+i}\hspace{0.5cm}\forall i\in\mathbb{I}_0^\infty
\end{equation}
where $\varepsilon_{t+i}\triangleq (\hat{A}_{t+i\color{black}{+1}}-F)\tilde{x}_{t+i}+(\hat{\psi}_{t+i\color{black}{+1}}-\hat{\psi}_t)\hat{Z}_{t+i}\in\mathbb{R}^n$, with $\hat{Z}_{t+i}\triangleq\begin{bmatrix}\hat{x}_{t+i}^T&u_{t+i}\end{bmatrix}^T\in\mathbb{R}^{n+1}$, is coined as the prediction uncertainty. All the terms involving the unavailable quantities at time $t$ in \eqref{epsil1} are lumped together in $\varepsilon_{t+i}$. This prediction uncertainty is handled using a homothetic tube \cite{langson2004robust}, as detailed in the subsequent subsections.

\subsection{Sets for the prediction uncertainties}
 By definition, $\varepsilon_{t+i}=(\hat{A}_{t+i\color{black}{+1}}-F)\tilde{x}_{t+i}+(\hat{\psi}_{t+i\color{black}{+1}}-\hat{\psi}_t)\hat{Z}_{t+i}$ $=(\hat{A}_{t+i\color{black}{+1}}-F)\tilde{x}_{t+i}+\sum_{k=t}^{\color{black}{t+i}}(\hat{\psi}_{k+1}-\hat{\psi}_k)\hat{Z}_{t+i}$ $\forall i\in\mathbb{I}_0^\infty$.
Thus,
\begin{equation}\label{vareS}
    \varepsilon_{t+i}\in\Omega_2\oplus i\Omega_3\hspace{0.7cm}\forall i\in\mathbb{I}_0^{\infty}
\end{equation}
where $\Omega_2\triangleq\{(\hat{A}-F)\tilde{x}\;|\;\hat{A}\in\Psi_A\text{, }\tilde{x}\in\mathbb{W}_{oa}\}\textcolor{black}{\oplus\;\Omega_3}$, $\Psi_A\triangleq$ $\{ \hat{A}\;|\begin{bmatrix} \hat{A}&\hat{b}\end{bmatrix}\in\Psi\text{ with }\hat{b}\in\mathbb{R}^{n}\}$ and $\Omega_3\triangleq \{\bar{\psi}\hat{Z}\;|\;\bar{\psi}\in\widetilde{\Psi}\text{; }\hat{Z}=[\;\hat{x}^T\;u\;]^T\text{ with }\hat{x}\in\widehat{\mathcal{X}}\text{, }u\in\mathcal{U}\}$. 

\subsection{Tubes for state estimate and control input}
 The optimization routine provides a state estimate tube $\mathbf{S}_t\triangleq\{S_{0|t},S_{1|t},$ $...,S_{N|t}\}$ and a control tube $\mathbf{U}_t\triangleq\{U_{0|t},U_{1|t},$ $...,U_{N-1|t}\}$ at each $t$. The state estimate tube sections are designed as \cite{langson2004robust}
\begin{equation}\label{st}
    S_{i|t}=\textcolor{black}{\{\beta_{i|t}\}\oplus}\;\alpha_{i|t}D\subseteq\widehat{\mathcal{X}}\hspace{0.7cm}\forall i\in\mathbb{I}_{0}^{N}
\end{equation}
where $D=conv(\{d^{v1},d^{v2},...,d^{vj}\})$ \big{[}$(\cdot)^{vk}$ denotes $k^{\text{th}}$ vertex\big{]} is a convex polytope containing the origin, and $\{\beta_{i|t}\}$ and $\{\alpha_{i|t}\}$ are the sequences of centers and scaling factors, respectively, for the tube sections in $\mathbf{S}_t$. From \eqref{st}, we express, $S_{i|t}=conv(\{s_{i|t}^{v1},s_{i|t}^{v2},...,s_{i|t}^{vj}\})$ with $s_{i|t}^{vk}=\beta_{i|t}+\alpha_{i|t}d^{vk}\text{ }\forall (i,k)\in$ $\mathbb{I}_{0}^{N}\times\mathbb{I}_1^j$. And the control tube sections are given by \cite{langson2004robust}\newpage\vspace*{-0.2cm}\noindent
\begin{equation}\label{ct}
   U_{i|t}=\{u_{i|t}^{v1},u_{i|t}^{v2},...,u_{i|t}^{vj}\}\subseteq\mathcal{U}\hspace{0.7cm}\forall i\in\mathbb{I}_{0}^{N-1}
\end{equation}
 where each $u_{i|t}^{vk}$ is linked to $s_{i|t}^{vk}\text{ }\forall (i,k)\in\mathbb{I}_{0}^{N-1}\times\mathbb{I}_1^j$, and is required to satisfy the inclusion given later in (\ref{cons6m}).
\subsection{Characterization of the Terminal Set}
For the AOFMPC, we make the following assumption.
\begin{assu}\label{Ka}
There exists a pair $(P,K)$, where $P\in\mathbb{R}^{n\times n}$ and $K^T\in\mathbb{R}^n$, that satisfies $P~\succ~0 \text{ and }-(\hat{A}+\hat{b}K)^TP(\hat{A}+\hat{b}K)+P-{\textcolor{black}{(Q+rK^TK)}}\succ 0$ $\forall\text{ }\begin{bmatrix}\hat{A}&\hat{b}\end{bmatrix}\in\Psi$.
In addition, $\exists$ a $\lambda$-contractive set $\widehat{\mathcal{X}}_T\triangleq \xi D\subseteq\widehat{\mathcal{X}}$ where $\xi\in\mathbb{R}$ such that \textcolor{black}{$(\hat{A}+\hat{b}K)\hat{x}\in\lambda\widehat{\mathcal{X}}_T$} $\forall (\hat{x}, K\hat{x})\in\widehat{\mathcal{X}}_T\times\mathcal{U} \text{ with } \lambda\in(0,1)$. 
\end{assu}

The terminal set for AOFMPC is  $\widehat{\mathcal{X}}_T$, inside which control input $u_t=K\hat{x}_t$ $\forall \hat{x}_t\in\widehat{\mathcal{X}}_T$. The constant $\lambda$ is chosen such that
\begin{equation}\label{eq29}
    \lambda\widehat{\mathcal{X}}_T\subseteq\widehat{\mathcal{X}}_T\ominus\{\Omega_2\oplus (N-1)\Omega_3\} \text{ (from \eqref{vareS})}.
\end{equation}
Assumption \ref{Ka} is standard in the context of adaptive MPC \cite{dhar2021indirect}. A set $\widehat{\mathcal{X}}_T$ satisfying \eqref{eq29} is computable following \cite{darup2016computation}.
\subsection{Reformulated COCP for AOFMPC}
 Let $\theta_t\triangleq\{\{\alpha_{i|t}\},\{\beta_{i|t}\},\mathbf{U}_{t}\}$ $\forall i\in\mathbb{I}_0^N$. The reformulated COCP for AOFMPC is given by
\begin{align}
  &\mathbb{OP}2:\min_{\theta_t}\;J(\hat{x}_t,\theta_t) \;\;\text{where}\nonumber\\
  &J(\hat{x}_t,\theta_t)\triangleq \sum_{i=0}^{N-1}\sum_{k=1}^j\Big{\{}||s^{vk}_{i|t}||^2_Q+r(u^{{vk}}_{i|t})^2\Big{\}}
  +\sum_{k=1}^{j}||s^{vk}_{N|t}||^2_P \label{MPC2}
\end{align} \vspace{-0.4cm}
\begin{align}
    &\text{subject to }\hspace{0.2cm}\beta_{0|t}=\hat{x}_t,\hspace{0.1cm}\alpha_{0|t}=0\text{ and }\alpha_{i|t}\geq0\text{ }\hspace{0.2cm}\forall i\in\mathbb{I}_{1}^{N}\tag{\ref{MPC2}a}\label{cons1m}\\
  &\hspace{2.5cm} S_{i|t}\subseteq\widehat{\mathcal{X}},\text{ }U_{i|t}\subseteq\mathcal{U}\hspace{0.3cm}\forall i\in\mathbb{I}_{0}^{N-1}\tag{\ref{MPC2}b}\label{cons2m}\\
   &\hspace{3.2cm} S_{N|t}\subseteq\widehat{\mathcal{X}}_T\subset\widehat{\mathcal{X}}\tag{\ref{MPC2}c}\label{cons3m}\\
    &\hspace{-0.162cm}\hat{A}_ts_{i|t}^{vk}+\hat{b}_tu_{i|t}^{vk}\in S_{i+1|t}\ominus\{\Omega_2\oplus i\Omega_3\}\hspace{0.1cm}\forall (i,k)\in\mathbb{I}_{0}^{N-1}\times\mathbb{I}_1^j\tag{\ref{MPC2}d}\label{cons6m}
\end{align}
where $P$ and $\widehat{\mathcal{X}}_T$ are defined following Assumption \ref{Ka} and \eqref{eq29}. From \eqref{cons1m}, $S_{0|t}$ and $U_{0|t}$ are singleton sets with respective elements $\hat{x}_t$ and $u_{0|t}\triangleq u_{0|t}^{v1}=u_{0|t}^{v2}=...=u_{0|t}^{vj}$. The control input applied to \eqref{osys1} and \eqref{AOdy} at time $t$ is $u_t= u_{0|t}$. The inclusion \eqref{cons6m} obtained from \eqref{epsil1}-\eqref{vareS} ensures $\hat{x}_{t+i}\in S_{i|t}$ $\forall (t,i)\in\mathbb{I}_0^\infty \times \mathbb{I}_0^N$. The complete homothetic tube containing the state estimate trajectory is constructed using $S_{0|0}\ni \hat{x}_0$ and $S_{1|t}\ni \hat{x}_{t+1}$ $\forall t\in\mathbb{I}_0^\infty$. Its pictorial representation is given in Fig. \ref{Introp} as the yellow \textit{inner tube}, with the state estimate trajectory in red. Adding the invariant set $\mathbb{W}_{oa}$ to the tube sections of the homothetic tube results in a larger \textit{outer tube} (Fig. \ref{Introp}). Since the state estimation errors belong to $\mathbb{W}_{oa}$, the actual state trajectory is guaranteed to lie inside the outer tube. Algorithm $1$ provides the steps to implement AOFMPC in a receding horizon fashion \cite{kouvaritakis2016model}.
\begin{remark}\textcolor{black}{
The size of the set for $\varepsilon_{t+i}$ given in \eqref{vareS} increases linearly with $i$, implying that the feasible region in $\mathcal{X}$ reduces with increasing horizon, as seen from \eqref{eq29} and \eqref{cons6m}. It is later shown in \ref{subF} that the horizon-dependent bound is crucial for proving recursive feasibility, albeit at the cost of increased conservatism.}
\end{remark}

\begin{figure}[t!]
 \vspace{0.23cm} \centering
     \framebox{\parbox{3in}{\includegraphics[scale=0.391]{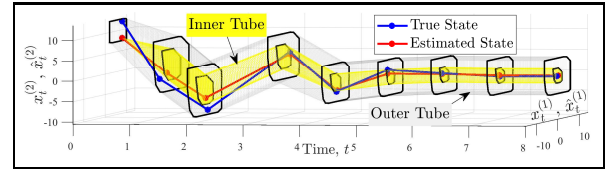}}}
     \caption{Pictorial representation of the tubes' structure for a $2^\text{nd}$ order SISO LTI system. Superscripts on $x_t$ and $\hat{x}_t$ denote the respective components.}
      \label{Introp}
\end{figure}
\begin{algorithm}
\caption{AOFMPC}\label{alg}
\begin{algorithmic}
\Offline Compute $\mathbb{W}_{oa}$ from \eqref{xtinW0}, $\widehat{\mathcal{X}}$ from \eqref{hatXtc}, $\Omega_2$, $\Omega_3$ from \eqref{vareS} and, $(P,K)$ for $\Psi$ and $\widehat{\mathcal{X}}_T$ from Assumption \ref{Ka} and \eqref{eq29}.\;
\Online{ 
\For{$t\geq 0$}\\
\begin{itemize}
    \item Measure $y_t$ from the plant \eqref{osys1}.
    \item Run $\mathbb{OP}2$ with $\hat{A}_t$, $\hat{b}_t$ and $\hat{x}_t$ to get $u_t$.
    \item Apply $u_t$ and $y_t$ to the observer \eqref{AOdy} to get $\hat{A}_{t+1}$, $\hat{b}_{t+1}$ and $\hat{x}_{t+1}$. Simultaneously, apply $u_t$ to the plant \eqref{osys1}.
    \item Update $t \gets t+1$.
\end{itemize} \EndFor}
\end{algorithmic}
\end{algorithm}
\begin{remark} The proposed AOFMPC allows imposition of user-defined bounds on the internal states, \textcolor{black}{which is not straightforward using the methods in \cite{coulson2019data}-\cite{nguyen2020output}, \cite{heirung2017dual}. Also, the adaptation in AOFMPC leads to parameter learning and} \newpage\vspace*{-0.6cm}\noindent{\textcolor{black}{\noindent therefore, reduced conservatism as compared to the non-adaptive approaches in \cite{mayne2006robust}, \cite{kogel2017robust}. In addition, the modularity of the adaptive observer design acts as a stepping stone for exploring other variants of robust MPC that do not rely on polytopic tubes, to further reduce the conservatism.}}
\end{remark}
\subsection{Recursive Feasibility and Boundedness}\label{subF}
 We make the following claims for $\mathbb{OP}2$.
\begin{lemma}\label{Tif1}
Suppose ${x}_0\in{\mathcal{X}}$, $\hat{x}_0\in\widehat{\mathcal{X}}$and $\mathbb{OP}2$ is initially feasible (i.e., solution of $\mathbb{OP}2$ exists at $t=0$) resulting in control $\mu_0$. Then, $\varepsilon_{i}\in\Omega_2\oplus i\Omega_3\hspace{0.2cm}\forall i\in\mathbb{I}_0^{N-1}$, and $\tilde{x}_{i}\in\mathbb{W}_{oa}$ guaranteeing $x_{i}\in\mathcal{X}\hspace{0.2cm}\forall i\in\mathbb{I}_1^N$.
\end{lemma} 
\begin{proof}
From \eqref{xtinW0}, $\tilde{x}_0=x_{0}-\hat{x}_{0}\in\mathcal{X}-\hat{\mathcal{X}}\Rightarrow \tilde{x}_0\in\mathbb{W}_{oa}$.\\ $\therefore$ $\varepsilon_{0}=(\hat{A}_1-F)\tilde{x}_{0}\textcolor{black}{+(\hat{\psi}_1-\hat{\psi}_0)\hat{Z}_0}\in\Omega_2$ (by definition of $\Omega_2$). Further, \eqref{xtinW0} and \eqref{cons6m} guarantee that $\tilde{x}_{1}\in\mathbb{W}_{oa}$ and  $\hat{x}_1\in S_{1|0}\subseteq\widehat{\mathcal{X}}$, respectively.

$\therefore\text{ }x_{1}=\hat{x}_{1}+\tilde{x}_{1}\in\widehat{\mathcal{X}}\oplus\mathbb{W}_{oa}\Rightarrow x_1\in\mathcal{X}\text{ (from \eqref{hatXtc})}.$

Using similar steps as done for $\varepsilon_{0}$, $\tilde{x}_{1}$ and $x_{1}$, it can be proved that $\varepsilon_{i}\in\Omega_2\oplus i\Omega_3\hspace{0.2cm}\forall i\in\mathbb{I}_1^{N-1}$ and $\tilde{x}_{i}\in\mathbb{W}_{oa}$ guaranteeing $x_{i}\in\mathcal{X}$ $\forall i\in\mathbb{I}_2^N$.
\end{proof}

\begin{coro}\label{cor1}
If $\forall t\in\mathbb{I}_0^\infty$, ${x}_t\in{\mathcal{X}}$, $\hat{x}_t\in\widehat{\mathcal{X}}$ and $\mathbb{OP}2$ is feasible at time $t$ resulting in control $\mu_t$, then, $\varepsilon_{t+i}\in\Omega_2\oplus i\Omega_3$ $\forall i\in\mathbb{I}_0^{N-1}$, and $\tilde{x}_{t+i}\in\mathbb{W}_{oa}$ guaranteeing $x_{t+i}\in\mathcal{X}$ $\forall i\in\mathbb{I}_{1}^{N}$.
\end{coro}

\begin{theo}\label{theo1}
If $\mathbb{OP}2$ is feasible at any time $t$, then it will be recursively feasible at all time $t+i\hspace{0.25cm}\forall i\in\mathbb{I}_1^\infty$.
\end{theo}
\begin{proof}
Given that $\mathbb{OP}2$ is feasible at time $t$, (\ref{cons1m})-(\ref{cons6m}) are satisfied. For time $t+1$, let a solution be proposed in terms of state and input tubes as $\mathbf{S}_{t+1}=\{S_{0|t+1},S_{1|t+1},...,S_{N|t+1}\}$ and $\mathbf{U}_{t+1}=\{U_{0|t+1},$ $U_{1|t+1},...,U_{N-1|t+1}\}$, where
    \begin{gather}\begin{aligned}\label{RF3}
   &S_{i|t+1}=S_{i+1|t}\hspace{0.2cm}\forall i\in\mathbb{I}_{1}^{N-1},\hspace{0.2cm}S_{N|t+1}=\widehat{\mathcal{X}}_T\\
    &U_{i|t+1}=U_{i+1|t}\hspace{0.3cm}\forall i\in\mathbb{I}_{1}^{N-2},\hspace{0.2cm}U_{N-1|t+1}=K\widehat{\mathcal{X}}_T.
\end{aligned}\end{gather}
The convex combination of the vertices of $S_{1|t}$ that results in $\hat{x}_{t+1}$, is used to get $u_{t+1}$ from $U_{1|t}$. Thus,  $\beta_{0|t+1}=\hat{x}_{t+1}$ and $\alpha_{0|t+1}=0$ implying $S_{0|t+1}=\{\hat{x}_{t+1}\}$ and $U_{0|t+1}=\{u_{t+1}\}$, i.e., $u_{0|t+1}^{v1}=u_{0|t+1}^{v2}...=u_{0|t+1}^{vj}=u_{t+1}$. 

Adding and subtracting $\hat{A}_{t}s_{i|t+1}^{vk}+\hat{b}_{t}u_{i|t+1}^{vk}$ to the LHS of
\eqref{cons6m} at time $t+1$, instead of $t$, $\forall (i,k)\in\mathbb{I}_{0}^{N-2}\times\mathbb{I}_1^j$, we get 
\begin{align*}\label{RF5}
     & \hat{A}_{t+1}s_{i|t+1}^{vk}+\hat{b}_{t+1}u_{i|t+1}^{vk}\\
       =&(\hat{\psi}_{t+1}-\hat{\psi}_t)\begin{bmatrix}{s_{i|t+1}^{vk}}^T & u_{i|t+1}^{vk}
       \end{bmatrix}^T+\hat{A}_{t}s_{i+1|t}^{vk}+\hat{b}_{t}u_{i+1|t}^{vk}\\
       \in&\text{ }\Omega_3\oplus[S_{i+2|t}\ominus\{\Omega_2\oplus (i+1)\Omega_3\}]=S_{i+1|t+1}\ominus(\Omega_2\oplus i\Omega_3).
 \end{align*}
From \eqref{cons3m} and \eqref{RF3}, $S_{N-1|t+1}=S_{N|t}\subseteq \widehat{\mathcal{X}}_T$ and $U_{N-1|t+1}=K\widehat{\mathcal{X}}_T$. Using Assumption \ref{Ka} and \eqref{eq29},  $\forall k\in\mathbb{I}_1^j$, 
\begin{align*}
&\hat{A}_{t+1}s_{N-1|t+1}^{vk}+\hat{b}_{t+1}u_{N-1|t+1}^{vk}=(\hat{A}_{t+1}+\hat{b}_{t+1}K)s_{N-1|t+1}^{vk}\\
\in\;\;&\lambda S_{N-1|t+1}=\lambda S_{N|t}\subseteq\lambda \widehat{\mathcal{X}}_T
\subseteq  S_{N|t+1}\ominus\{\Omega_2\oplus (N-1)\Omega_3\}.
\end{align*}
%

Thus, $\mathbb{OP}2$ is feasible at $t+1$ with the proposed solution. Similarly, it can be proved that $\mathbb{OP}2$ will be feasible at $t+2$ using the solution for $t+1$, and one can recursively continue to prove recursive feasibility at all time $t+i$ $\forall i\in\mathbb{I}_1^\infty$.
\end{proof}

\begin{coro}\label{cor2}
The application of control $u_t=u_{0|t}$ ensures $x_t\in\mathcal{X}$ and $u_t\in U_{0|t}\subseteq\mathcal{U}$ $\forall t\in\mathbb{I}_0^\infty$. (\textit{Proof}: Follows from Lemma \ref{Tif1}, Corollary \ref{cor1}, and Theorem \ref{theo1}. Note: the control input constraint has not been modified in this paper).
\end{coro}

The recursive feasibility of $\mathbb{OP}2$ implies $\hat{x}_t$ enters $\widehat{\mathcal{X}}_T$ in at most $N$ steps. From \eqref{RF3}, the control inputs being implemented are $u_t=u_{0|t}$ $\forall t\in\mathbb{I}_0^{N-1}$ and $u_t=K\hat{x}_{t}$ $\forall t\in\mathbb{I}_N^\infty$. Using Assumption \ref{Ka}, it follows that $\hat{x}_t$ exponentially converges to the origin if \eqref{epsil1} is used for prediction. Since \eqref{epsil2} replaced \eqref{epsil1}, we arrive at two possible scenarios.

\begin{itemize}
    \item \textit{$\varepsilon_{t+i}$ converges to zero}: This requires $\tilde{p}_t$, and consequently, $\tilde{x}_t$ converging to zero following \eqref{xtilde}. By Remarks \ref{ronuSR} and \ref{L1}, $\varepsilon_{t+i}$ converges to zero if the input is sufficiently rich, \textcolor{black}{which can be achieved following \cite[inequality~(29)]{lorenzen2019robust}}. With such $u_t$, $\lim_{t\rightarrow\infty}\hat{x}_t=x_t$. However, the excitation constraint prevents $\hat{x}_t$ and $x_t$ from settling at zero even as $t\rightarrow\infty$.  
    \item \textit{$\varepsilon_{t+i}$ does not converge to zero}: In absence of any excitation, there is no guarantee that $\tilde{p}_t$ and hence, $\tilde{x}_t$ will converge to zero. At best, it can be guaranteed that $x_t\in\hat{x}_t\oplus\mathbb{W}_{oa}$.
\end{itemize}

In either case, recursive feasibility of $\mathbb{OP}2$ guarantees that $\hat{x}_t, u_t\in \mathcal{L}_\infty$ $\forall t\in\mathbb{I}_0^\infty$. By definition of $\widehat{\mathcal{X}}$, $\hat{x}_t\in \mathcal{L}_\infty\Rightarrow x_t\in \mathcal{L}_\infty$ $\forall t\in\mathbb{I}_0^\infty$. Additionally, \eqref{Proj} ensures $\hat{A}_t,\hat{b}_t\in \mathcal{L}_\infty$ $\forall t\in\mathbb{I}_0^\infty$.

\section{Numerical Example}
We consider a $2^{\text{nd}}$ order LTI system\footnote{\textcolor{black}{A $2^{\text{nd}}$ order example is chosen for ease of visualization of the tubes, although implementation on higher order systems can also be achieved, albeit at a higher computational cost.}}
\begin{align*}
   x_{t+1}= \begin{bmatrix}-0.6273 & 1\\ 0.4564 & 0
   \end{bmatrix}x_t+\begin{bmatrix}-0.1818\\0.0909
   \end{bmatrix}u_t\hspace{0.1cm}\text{ ; }\hspace{0.1cm}y_t=\begin{bmatrix}1&0
   \end{bmatrix}x_t.
\end{align*}
The set $\Psi=conv(\{\psi^{vi}\;|\;i\in\mathbb{I}_1^3\})$ with $\psi^{v1}=[-0.7\hspace{0.2cm} 1\hspace{0.2cm} -0.2\; ;$\\ $0.5\hspace{0.2cm} 0 \hspace{0.2cm}0.1]$, $\psi^{v2}=[-0.2\hspace{0.2cm} 1\hspace{0.2cm} 0.1\; ;-0.08\hspace{0.2cm} 0 \hspace{0.2cm}0]$ and $\psi^{v3}=$ $[0.1\hspace{0.2cm}1\hspace{0.2cm} 0\; ;0.02\hspace{0.2cm} 0 \hspace{0.2cm}0]$. The initial conditions are $\hat{\psi}_0=[0.092$\newpage\noindent $1\hspace{0.2cm}-0.002;\;0.0248$ $0\hspace{0.2cm}0.001]$, $x_0=[30;\;19]$, $\hat{x}_0=[25;\;15]$ and $\mathbb{W}_0=conv(\{[-1.9951;\;7.0747]{,}\;[5.0796;\;4.6056],$ $[8.0922;-6.0771],\;[1.9951;-7.0747],\;[-8.0922;6.0771],$ $[-5.0796;-4.6056] \})$. The state and input constraints are: $||x_t||_\infty\leq38$ and $||u_t||_\infty\leq45.5$. The other parameters for AOFMPC are: $N=7$, $\Gamma_0=100I_4$, $\zeta=0.25$, $Q=I_2$, $r=0.1$ 
and $F=[0.49$ $1;\;-0.5\hspace{0.2cm}0]$. The simulations have been carried out using \cite{MPT3} and \cite{Lofberg2004}.

    \begin{figure}[t!]
     \vspace{0.23cm} \centering
     \framebox{\parbox{3in}{\includegraphics[scale=0.4]{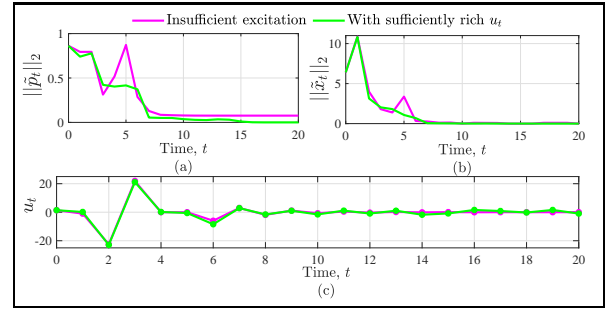}}}
     \caption{(a) $2$-norm of parameter estimation error, (b) $2$-norm of state estimation error, in absence and presence of sufficient excitation, and (c) corresponding input.} 
      \label{ParaE}
   \end{figure}
   \begin{figure}[t!]
      \centering
     \framebox{\parbox{3in}{\includegraphics[scale=0.4]{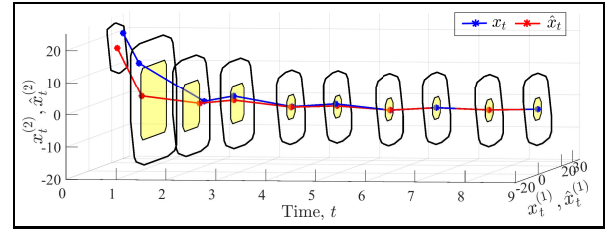}}}
      \caption{State tubes for $N=7$. The inner (yellow) tube and the outer (combined yellow and white) tubes correspond to the case when $u_t$ is not sufficiently rich. Superscripts on $x_t$ and $\hat{x}_t$ denote the respective components.} 
      \label{STube}
   \end{figure}
   \begin{table}[t!]\caption{Comparison of AOFMPC and \cite{dhar2021indirect}.}\vspace{-0.3cm}
    \label{table1}\begin{center}
    \resizebox{\columnwidth}{!}{
     \begin{tabular}{|c|c|c|c|c|c|c|}
\hline
& \multicolumn{2} {|c|}{\textit{From $t=0$ to $3$ }} & \multicolumn{2}{|c|}{\textit{From $t=0$ to $15$}}
& \multicolumn{2}{|c|}{\textit{From $t=0$ to $30$}}
\\
&\textbf{ AOFMPC}&\textbf{\cite{dhar2021indirect}}& \textbf{ AOFMPC}&\textbf{\cite{dhar2021indirect}}& \textbf{ AOFMPC}&\textbf{\cite{dhar2021indirect}}\\ \hline
  
  
  
  
  
  

 {\scriptsize{RMS}} $||x_t||_2$&20.4261 & 18.7163& 10.2985&9.6443 & 7.3987& 6.9004\\ \hline
  
  {\scriptsize{RMS}} $||\tilde{p}_t||_2$& 0.7224& 0.8598& 0.4469& 0.6588& 0.3254& 0.6584\\ \hline
  
{$\sum_{\hspace{0.05cm}t}\;\mathbb{J}_t$} & 2569.4& 2101.2&2737.3 & 2198.2& 2737.5& 2199.1 \\ \hline

    \end{tabular}
   }
    \end{center}\end{table} 
 
Figs. \ref{ParaE}(a) and (b) demonstrate two cases,- without sufficient excitation, $||\tilde{p}_t||_2$ and $||\tilde{x}_t||_2$ reach the vicinity of zero, whereas when $u_t$ is sufficiently rich\footnote{\textcolor{black}{Obtained by adding a dithering signal of amplitude $0.5$ consisting of sinusoidals of five co-prime frequencies, that satisfy the excitation condition.}}, both $||\tilde{p}_t||_2$ and $||\tilde{x}_t||_2$ converge to zero. The implemented controls are shown in Fig. \ref{ParaE}(c). Fig. \ref{STube} shows the tubes for the case when $u_t$ is not sufficiently rich. The inner tube (in yellow) is for $\hat{x}_t$ (in red line); the blue line is for $x_t$ that evolves inside the outer tube. 

\textcolor{black}{A comparison of AOFMPC with \cite{dhar2021indirect} is made in Table \ref{table1}} (extra parameters defined in \cite{dhar2021indirect} are $\lambda=5.6\times10^{-4}$ and $\alpha=0.005$). The root mean square (RMS) value of $||x_t||_2$ is smaller in \cite{dhar2021indirect} during the initial time instants but is comparable with AOFMPC in the steady state. Parameter estimation is better with AOFMPC; heuristically, this may be linked to the use of RLS in AOFMPC and gradient descent in \cite{dhar2021indirect}. The value of $\sum_t\mathbb{J}_t\triangleq x_t^Tx_t+u_t^2$ is higher in AOFMPC; this is expected since AOFMPC does not use state measurements and instead utilizes estimates generated from the observer. On the other hand, \cite{dhar2021indirect} uses full state measurements and incurs a smaller cost. 


\textcolor{black}{Two more cases are illustrated},- (i) with measurement noise $y_{n_t}$ (where $y_t=cx_t+y_{n_t}$), and (ii) with disturbance $x_{d_t}$ (where $x_{t+1}=Ax_t+bu_t+x_{d_t}$). The signals $y_{n_t}$ and $x_{d_t}$ are randomly generated with amplitude between $0$ and $1$. Figs. \ref{ParaN}(a)-(d) show plots of $||\tilde{p}_t||_2$, $||\tilde{x}_t||_2$, $||x_t||_2$ and $u_t$, respectively, with $x_0=[5;\;-1]$ and $\hat{x}_0=[4;\;3]$. AOFMPC is feasible in such cases with noise/disturbance albeit with a smaller initial feasible region, which is a trade-off with robust guarantees.
 \begin{figure}[t!]
   \vspace{0.23cm}
    \centering
     \framebox{\parbox{3in}{\includegraphics[scale=0.38]{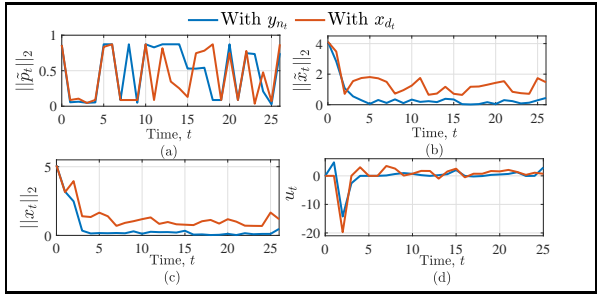}}}
     \caption{(a) $2$-norm of parameter estimation error, (b) $2$-norm of state estimation error, (c) $2$-norm of state and (d) implemented input in presence of noise $y_{n_t}$ and disturbance $x_{d_t}$. }
      \label{ParaN}
  \end{figure} 
\section{Conclusion}
The paper proposes a technique to solve the COCP for uncertain discrete-time SISO LTI systems using only output measurements. The solution approach involves an MPC using estimates of the states and system parameters, which are obtained simultaneously from an adaptive observer at each time instant. Reformulating the COCP using the online available estimates leads to the introduction of uncertainties in the state estimate predictions, and is tackled using a homothetic tube. Additionally, an invariant set for the state estimation error is characterized. The optimization routine ensures that the state estimates are in a constraint set tightened by the invariant set. Adding the invariant set to the homothetic tube sections creates a larger tube that contains the actual state trajectory. The two-tube architecture ensures that the hard constraint on the actual state is never violated for any possible value of the state estimation error. Simulation results show the performance of the proposed AOFMPC with and without a sufficiently rich input. A detailed analysis of stability of the proposed theory will be done as a part of future work. Some immediate extensions of this work is to design AOFMPC for multi-input multi-output systems and \textcolor{black}{for cases of noisy measurements and/or external disturbances,} \textcolor{black}{and to exploit the adaptation for gaining better knowledge of the error bound and thus, reduce conservatism.}





\bibliographystyle{IEEEtran}
\bibliography{IEEEabrv,reference}

\end{document}